\documentclass[12pt]{article}
\usepackage{amscd,amsmath,amssymb,amstext,amsthm,exscale,latexsym}
\usepackage{array,graphicx}
\newcommand {\al}   {\alpha}       \newcommand {\bt}  {\beta}
\newcommand {\g }   {\gamma}       \newcommand {\G }  {\Gamma}
\newcommand {\dl}   {\delta}       \newcommand {\e }  {\epsilon}

\newcommand {\lm}   {\lambda}      
          
\newcommand {\s }   {\sigma}       
         
\newcommand {\vf }  {\varphi}

\newcommand {\pl}   {\partial}     \newcommand {\nb}  {\nabla}
\renewcommand {\sin}{{\sf\,sin\,}}       
           
\renewcommand {\det}{{\sf\,det\,}}

\newcommand   {\re}{{\sf\,re\,}}         
\newcommand   {\sign}{{\sf\,sign\,}}     
\newcommand   {\const}{{\sf\,const}}     \newcommand   {\diag}{{\sf\,diag\,}}
   
\newcommand {\MC}  {{\mathbb C}}

\newcommand {\MM}  {{\mathbb M}}   
   
   \newcommand {\MR}  {{\mathbb R}}

\newcommand {\Bq}  {\boldsymbol{q}}   \newcommand {\Bx}  {\boldsymbol{x}}
      
      \renewcommand {\CD}  {{\cal D}}

\newcommand {\vol}  {\sqrt{|g|}}
\newtheorem{theorem}{Theorem}[section]
\begin{document}
\title     {Point massive particle in General Relativity}
\author    {M. O. Katanaev
            \thanks{E-mail: katanaev@mi.ras.ru}\\ \\
            \sl Steklov Mathematical Institute,\\
            \sl ul. Gubkina, 8, Moscow, 119991, Russia}
\maketitle
\begin{abstract}
It is well known that the Schwarzschild solution describes the gravitational
field outside compact spherically symmetric mass distribution in General
Relativity. In particular, it describes the gravitational field outside a point
particle. Nevertheless, what is the exact solution of Einstein's equations with
$\dl$-type source corresponding to a point particle is not known. In the present
paper, we prove that the Schwarzschild solution in isotropic coordinates is the
asymptotically flat static spherically symmetric solution of Einstein's
equations with $\dl$-type energy-momentum tensor corresponding to a point
particle. Solution of Einstein's equations is understood in the generalized
sense after integration with a test function. Metric components are locally
integrable functions for which nonlinear Einstein's equations are mathematically
defined. The Schwarzschild solution in isotropic coordinates is locally
isometric to the Schwarzschild solution in Schwarzschild coordinates but differs
essentially globally. It is topologically trivial neglecting the world line of a
point particle. Gravity attraction at large distances is replaced by repulsion
at the particle neighborhood.
\end{abstract}
\section{Introduction}
In this article, we consider the classic problem: find the gravitational field
which is produced by a point massive particle. If particle is at rest, then
the gravitational field is spherically symmetric and static. The spherically
symmetric solution of the vacuum Einstein's equations is well known: it is the
Schwarzschild solution \cite{Schwar16}. Therefore, it is often stated that the
Schwarzschild solution (in the Schwarzschild coordinates) describes
gravitational field of a point particle. This statement is incorrect because
there is no $\dl$-type energy-momentum tensor corresponding to a particle on the
right hand side of Einstein's equations.

On the other hand, the solution of Einstein's equations outside point massive
particle must be isometric to the Schwarzschild solution. Therefore the natural
question arises: ``Where is the $\dl$-function ?''. The answer turned out to be
unexpected: $\dl$-function corresponds to infinite value of the Schwarzschild
radial coordinate. Namely, we prove in this paper, that the Schwarzschild
solution in isotropic coordinates is the solution of Einstein's equations in a
topologically trivial space-time $\MR^4$ with $\dl$-type source. The solution is
understood in the generalized sense after integration with a test function. At
the same time, the metric components are locally integrable functions for which
the nonlinear Einstein's equations are mathematically defined.

Note that the Schwarzschild solution in Schwarzschild coordinates is locally
isometric to the Schwarzschild solution in isotropic coordinates. Nevertheless
global structure of space-time is totally different. The maximal extension of
space-time along geodesics for the Schwarzschild solution in Schwarzschild
coordinates is the topologically nontrivial manifold which is equal to the
topological product of the well known Carter--Penrose diagram with the sphere.
This maximally extended solution describes white and black holes and is not
related to the gravitational field of a point particle. At the same time the
Schwarzschild solution in isotropic coordinates corresponds to the topologically
trivial space-time.

Severe mathematical difficulties arise during solution of this problem. A
solution of Einstein's equations must be understood in a generalized sense after
integration with test functions because the $\dl$-function stands on the right
hand side. But there is no multiplication in the space of generalized functions
(distributions) and the question arises what is the mathematical meaning of the
left hand side of Einstein's equations which are nonlinear. Besides, locally
nonintegrable functions arise during solution of equations, and some functionals
must be attributed to them. In other words, regularization is needed. In the
preset paper, the solution of Einstein's equations is found in a generalized
sense: the equations are satisfied after integration with test functions. We
choose the usual space $\CD(\MR^3)$ of infinitely differentiable functions on
$\MR^3$ with compact support as the space of test functions. Metric components
for the obtained solution are locally integrable functions and therefore belong
to the conjugate space $\CD'(\MR^3)$. Though no multiplication is defined in
$\CD'$, the left hand side of Einstein's equations is well defined for the
derived solution. Here we use the analytic regularization for exponential
spherically symmetric functionals.

The obtained solution turned out to be the well known Schwarzschild solution in
isotropic coordinates. It appears after gluing together two exterior solutions
of a black hole along the horizon. This solution is isometric to Einstein--Rosen
bridge and is asymptotically flat not only at large distances but near the
particle itself where the curvature of space-time tends also to zero
(gravitational asymptotic freedom). At large distances, the gravitational field
is attractive. Under the horizon attraction changes to repulsion. This repulsion
results in geodesic completeness of space-time near the particle.

Attempts to interpret the Schwarzschild solution in terms of generalized
functions were made earlier [2--6]. Papers [2--4] are related to our approach
and discussed in some detail in Section \ref{seldef}. In reference
\cite{Tanghe61}, the energy-momentum tensor of matter was taken to have
spherically symmetric Gaussian distribution around the origin of the coordinate
system, and exact solution to Einstein's equation is found. In the limit of zero
distribution radius, the energy-momentum tensor is proportional to the
$\delta$-function. This energy-momentum tensor differs from that in our paper
and does not correspond to a point particle. Another approach was adopted in
\cite{Bel69}. Roughly speaking, the Schwarzschild sphere was shrank to a point.
The energy-momentum tensor was shown to be proportional to the $\delta$-function
at this point in a sense of a distribution, though not all Einstein's equations
were actually solved. Again, the energy-momentum tensor is different from ours.
The authors of \cite{BalNac93} regularized the Schwarzschild metric in
Schwarzschild coordinates. Afterwards they calculated the energy momentum tensor
and take off the regularization. The resulting energy-momentum tensor is
proportional to the $\delta$-function located at the origin. In this case, the
energy-momentum tensor also differs from ours.
\section{Point mass in General Relativity}
Let us consider topologically trivial manifold $\MM\approx\MR^4$ (space-time)
with Cartesian coordinates $x^\al$, $\al=0,1,2,3$, and metric $g_{\al\bt}(x)$ of
Lorentzian signature $\sign g_{\al\bt}=(+---)$. We denote the world line of a
point particle by $\lbrace q^\al(\tau)\rbrace$ where $\tau$ is a parameter along
the world line. We use the following notations for geometric notions:
\begin{align}                                                     \label{qrhjff}
  \G_{\al\bt}{}^\g&:=\frac12g^{\g\dl}(\pl_\al g_{\bt\dl}
  +\pl_\bt g_{\al\dl}-\pl_\dl g_{\al\bt}),
\\                                                                \label{qcuyte}
  R_{\al\bt\g}{}^\dl&:=\pl_\al\G_{\bt\g}{}^\dl-\pl_\bt\G_{\al\g}{}^\dl
  -\G_{\al\g}{}^\e\G_{\bt\e}{}^\dl+\G_{\bt\g}{}^\e\G_{\al\e}{}^\dl,
\\                                                                \label{qricci}
  R_{\al\bt}&:=R_{\al\g\bt}{}^\g,
\\                                                                \label{qcsafr}
  R&:=g^{\al\bt}R_{\al\bt},
\end{align}
where $\G_{\al\bt}{}^\g$ are Christoffel's symbols, $R_{\al\bt\g}{}^\dl$ is
the curvature tensor, $R_{\al\bt}$ is the Ricci tensor, and $R$ is the scalar
curvature.

In General Relativity, a point particle of mass $M$ is described by the
following action
\begin{equation}                                                  \label{eacwek}
  S=\frac1{16\pi}\int\!dx\vol R
  -M\int\!d\tau\sqrt{\dot q^\al\dot q^\bt g_{\al\bt}},
\end{equation}
where $g:=\det g_{\al\bt}$ and $\dot q:=dq/d\tau$.

Variations of action (\ref{eacwek}) with respect to metric components
$g_{\al\bt}$ and coordinates of a particle $q^\al$ yield Einstein's equations
of motion and equations for extremals (geodesics):
\begin{align}                                                     \label{eiepop}
  R^{\al\bt}-\frac12g^{\al\bt}R=-\frac12T^{\al\bt}&,
\\                                                                \label{extwel}
  \left(\ddot q^\al+\left.\G_{\bt\g}{}^\al\right|_{x=q}\dot q^\bt\dot q^\g
  \right)g_{\al\dl}=0&,
\end{align}
where
\begin{equation}                                                  \label{emotej}
  T^{\al\bt}=\frac{16\pi M\dot q^\al\dot q^\bt}{\vol\dot q^0}\dl(\Bx-\Bq)
\end{equation}
is the particle energy-momentum tensor, $\dot q_\al:=\dot q^\bt g_{\bt\al}$, and
\begin{equation*}
  \dl(\Bx-\Bq):=\dl(x^1-q^1)\dl(x^2-q^2)\dl(x^3-q^3)
\end{equation*}
is the three-dimensional $\dl$-function on sections $x^0=\const$. We assume
that for particle trajectory $\dot q^0\ne0$.

To analyze system of equations (\ref{eiepop})--(\ref{extwel}), we rewrite it in
the Hamiltonian form. To this end, the ADM parameterization of the metric is
used
\begin{equation}                                                  \label{eadmml}
  g_{\al\bt}=\begin{pmatrix} N^2+N^\rho N_\rho & N_\nu \\
           N_\mu & g_{\mu\nu}\end{pmatrix},
\end{equation}
where $g_{\mu\nu}$ is the metric on space like sections of the space-time
$x^0=\const$. In this parameterization, $4$ functions: the lapse function $N(x)$
and shift functions $N_\mu(x)$, are introduced instead $4$ components $g_{00}$
and $g_{0\mu}$ of the metric. Here we use notation
$N^\rho:=\hat g^{\rho\mu}N_\mu$ where $\hat g^{\rho\mu}$ is the $3\times3$
matrix which is inverse to $g_{\mu\nu}$:
\begin{equation*}
  \hat g^{\rho\mu}g_{\mu\nu}=\dl^\rho_\nu.
\end{equation*}
In what follows, raising of space indexes which are denoted by Greek letters
from the middle of the alphabet ($\mu,\nu,\dotsc=1,2,3$) is performed using the
inverse three-dimensional metric marked with a hat. We assume that all sections
$x^0=\const$ are space like, and therefore the metric $g_{\mu\nu}$ is negative
definite.

Let $p^{\mu\nu}$ and $p_\al$ be momenta conjugate to generalized coordinates
$g_{\mu\nu}$ and $q^\al$. The action for a point particle is invariant with
respect to reparameterization of the world line. To simplify formulae, we fix
the gauge $\tau=q^0$. Then Einstein's equations (\ref{eiepop}) in the
Hamiltonian form reduce to constraint and dynamical equations. Constraints have
the form
\begin{align}                                                     \label{edycgw}
  H_\bot&=\frac1{\hat e}\left(p^{\mu\nu}p_{\mu\nu}-p^2\right)
  -\hat e\hat R+\sqrt{M^2+\hat p^2}\,\dl(\Bx-\Bq)=0,
\\                                                                \label{ekicop}
  H_\mu&=-2\hat\nb_\nu p^\nu{}_\mu-p_\mu\dl(\Bx-\Bq)=0,
\end{align}
where $\hat e:=\sqrt{|\det g_{\mu\nu}|}$,
$\hat p^2:=\hat g^{\mu\nu}p_\mu p_\nu$, $\hat R$ is the three dimensional
scalar curvature for metric $g_{\mu\nu}$, and $\hat\nb_\mu$ is the three
dimensional covariant derivative. Constraints (\ref{edycgw}) and (\ref{ekicop})
are called dynamical and kinematical, respectively. Dynamical Hamiltonian
Einstein's equations are
\begin{align}                                                     \label{edefmk}
  \dot g_{\mu\nu}&=\frac{2N}{\hat e}p_{\mu\nu}
  -\frac{N}{\hat e}g_{\mu\nu}p+\hat\nb_\mu N_\nu+\hat\nb_\nu N_\mu,
\\                                                                     \nonumber
  \dot p^{\mu\nu}&
  =\frac N{2\hat e}\hat g^{\mu\nu}\left(p^{\rho\s}p_{\rho\s}-\frac12p^2\right)
  -\frac{2N}{\hat e}\left(p^{\mu\rho}p^\nu{}_\rho-\frac12p^{\mu\nu}p\right)
  +\hat e(\hat\triangle N\hat g^{\mu\nu}-\hat\nb^\mu\hat\nb^\nu N)-
\\                                                                \label{ecamod}
  &-\hat eN\left(\hat R^{\mu\nu}-\frac12g^{\mu\nu}\hat R\right)
  -p^{\mu\rho}\hat\nb_\rho N^\nu-p^{\nu\rho}\hat\nb_\rho N^\mu
  +\hat\nb_\rho(N^\rho p^{\mu\nu})
  -\frac{Np^\mu p^\nu}{2\sqrt{{M^2}+\hat p^2}}\dl(\Bx-\Bq),
\end{align}
where $\hat\triangle:=\hat\nb^\mu\hat\nb_\mu$ is the three-dimensional
Laplace--Beltrami operator, $p:=p^{\mu\nu}g_{\mu\nu}$, and
$p^\mu:=\hat g^{\mu\nu}p_\nu$. Hamiltonian equations for geodesic have the form
\begin{align}                                                     \label{eqmoqf}
  \dot q^\mu&=-\left.\frac N{\sqrt{M^2+\hat p^2}}\right|_{\Bx=\Bq}\!\!\!p^\mu
  -\left.N^\mu\right|_{\Bx=\Bq},
\\                                                                \label{eqmopg}
  \dot p_\mu&=-\pl_\mu\left[N\sqrt{M^2+\hat p^2}
  -N^\nu p_\nu\right]_{\Bx=\Bq}.
\end{align}

Transition from Lagrangian equations of motion (\ref{eiepop}), (\ref{extwel}) to
Hamiltonian (\ref{edycgw})--(\ref{eqmopg}) is complicated. Hamiltonian
formulation of General Relativity was given in \cite{Dirac58B,ArDeMi62}.
For General Relativity and  a point particle it can be found i.e.\ in
\cite{GitTyu90}. Combined Hamiltonian system of equations of motion for a point
particle in General Relativity was considered in \cite{ArDeMi60B,MenSem00}.

Now we are prepared to solve the full system of equations
(\ref{edycgw})--(\ref{eqmopg}) for a point particle at rest.
We assume that the particle is located at the origin of the coordinate system:
\begin{equation*}
  q^0=x^0=\tau,\qquad q^\mu=0,~~\mu=1,2,3.
\end{equation*}
We are seeking spherically symmetric static solution of the system of equations
of motion (\ref{edycgw})--(\ref{eqmopg}). In this case, the coordinate system
can be chosen in such a way that shift functions are zero, $N_\mu=0$. We choose
also the Weyl flat gauge for the spatial metric
\begin{equation}                                                  \label{qfsdok}
  g_{\mu\nu}=-f^2\dl_{\mu\nu},
\end{equation}
where $\dl_{\mu\nu}:=\diag(+++)$ is the Euclidean metric and
$f=f\big(\sqrt{(x^1)^2+(x^2)^2+(x^3)^2}\big)$. In the static case, all momenta
vanish: $p^{\mu\nu}=0$ and $p_\mu=0$. Derivatives $\dot g_{\mu\nu}=0$ are also
zero. Thus there are two unknown functions in this case which depend only on
radius in spherical coordinates: $N(r)$ and $f(r)$.

If the above assumptions are fulfilled then equations (\ref{ekicop}),
(\ref{edefmk}), and (\ref{eqmoqf}) are identically satisfied. Equations
(\ref{edycgw}), (\ref{ecamod}), and (\ref{eqmopg}) take the form
\begin{align}                                                     \label{qdytwv}
  -\hat e\hat R+16\pi M\dl(\Bx)&=0,
\\                                                                \label{qsemby}
  \hat e\left(\hat\triangle N\hat g^{\mu\nu}-\hat\nb^\mu\hat\nb^\nu N\right)
  -\hat eN\left(\hat R^{\mu\nu}-\frac12g^{\mu\nu}\hat R\right)&=0,
\\                                                                \label{qlpojg}
  \pl_\mu N\big|_{\Bx=0}&=0,
\end{align}
where we divided equation (\ref{eqmopg}) by $M$.

It is important that equation (\ref{qsemby}) can not be divided by $\hat e$, and
indices $\mu,\nu$ can not be lowered because metric have singularity at $r=0$,
and the Ricci tensor and scalar curvature contain the $\dl$-function.
\section{The main equation                                       \label{seldef}}
Equation (\ref{qdytwv}) is the main equation which is to be solved.
The three-dimensional volume element $\hat e$ on the left hand side is needed
for covariance because $\dl$-function is the scalar density with respect to
coordinate transformations.

This equation is the simplest covariant second order differential equation
for metric with the $\dl$-function on the right hand side. In two dimensions in
Weyl flat gauge, it reduces to the Poisson equation and yields the fundamental
solution for two-dimensional Laplace operator. In higher dimensions it becomes
nonlinear, and its solution is not known up to now. We believe that this
equation may have its own applications in differential geometry which are not
discussed here.

We are seeking spherically symmetric solutions of equation (\ref{qdytwv}), and
that is in agreement with the symmetry of the right hand side. Any spherically
symmetric metric is Weyl flat, which means that there exists the coordinate
system where the metric has the form (\ref{qfsdok}). Then main equation
(\ref{qdytwv}) for the Weyl flat spherically symmetric metric takes the form
\begin{equation}                                                  \label{qfilht}
  \triangle f-\frac{\pl f^2}{2f}=-4\pi M\dl(\Bx),
\end{equation}
where $\triangle:=\pl_1^2+\pl_2^2+\pl_3^2$ is the flat Laplace operator and we
introduced notation $\pl f^2:=\dl^{\mu\nu}\pl_\mu f\pl_\nu f$.

Equation (\ref{qfilht}) is nonlinear because of the second term on the
left hand side. Solution of this equation should be understood in the
generalized sense after integration with a test function (see, i.e.\
\cite{Vladim71}) because there is the $\dl$-function on the right hand side. If
we look for a solution in the space of functionals (generalized functions)
$\CD'(\MR^3)$ then two serious problems arise. Since the equation is nonlinear,
the generalized functions must be multiplied, but multiplication in $\CD'$ is
absent. Second, we shall see in what follows that some terms in equation
(\ref{qfilht}) are not locally integrable, and hence a regularization of the
appearing integrals is necessary. Though multiplication of functionals in
$\CD'$ is not defined in general, for some particular functionals, i.e.\
sufficiently smooth, the left hand side has nevertheless definite meaning.

We look for spherically symmetric solution of equation (\ref{qfilht}) in the
spherical coordinate system $r,\theta,\vf$.
\begin{theorem}                                                   \label{toyfrd}
The function
\begin{equation}                                                  \label{qsolco}
  f=1+\frac Mr+\frac{M^2}{4r^2}=\left(1+\frac M{2r}\right)^2
\end{equation}
satisfies equation (\ref{qfilht}) which solution is understood in the
generalized sense after integration with a test function.
\end{theorem}
\begin{proof}
We use the analytic regularization for the generalized function $r^\lm$,
$\lm\in\MC$ \cite{GelShi58A}. For $\re\lm>-3$ the function $r^\lm$ is locally
integrable, and the functional
\begin{equation*}
  (r^\lm,\vf):=\int_{\MR^3}\!\!\!dx\, r^\lm\vf
\end{equation*}
is defined for all test functions $\vf\in\CD(\MR^3)$ (smooth functions with
compact support). This functional is
analytically continued on the whole complex plane $\lm\in\MC$ except simple
poles located on the real line $\lm=-3,-5,-7,\dotsc$ \cite{GelShi58A}. The poles
does not matter because we do not fall into them in the considered case.

Consider the functional
\begin{equation*}
  f_\nu:=\left(1+\frac M2r^\nu\right)^2,\quad \re\nu\ge1,\quad\in\CD'(\MR^3),
\end{equation*}
depending on parameter $\nu\in\MC$. Simple calculations show that the following
identity holds
\begin{equation}                                                  \label{qjghui}
  \frac{\pl f_\nu^2}{2f_\nu}=\frac{f_\nu^{\prime2}}{2f_\nu}
  =\frac{M^2}2\nu^2r^{2\nu-2},
\end{equation}
where prime denotes differentiation with respect to the radius. It is clearly
well defined for $\re\nu\ge1$. This identity can be analytically continued to
$\nu=-1$ because we do not fall into the poles described above. As a result, we
get the following expression for the second term on the left hand side of
equation (\ref{qfilht})
\begin{equation}                                                  \label{qcectr}
  \frac{\pl f^2}{2f}=\frac{M^2}{2r^4},
\end{equation}
This equality is understood in the generalized sense after integration with test
functions. Therefore equation (\ref{qfilht}) for the generalized function
(\ref{qsolco}) takes the form
\begin{equation*}
  \triangle\left(1+\frac Mr+\frac{M^2}{4r^2}\right)-\frac{M^2}{2r^4}
  =-4\pi M\dl(\Bx).
\end{equation*}

The function $M/r$ is the unique fundamental solution of the Laplace equation
decreasing at infinity (see., for example, \cite{Vladim71}):
\begin{equation}                                                  \label{qfulap}
  \triangle\left(1+\frac Mr\right)=\triangle\frac Mr=-4\pi M\dl(\Bx).
\end{equation}

The function $M^2/2r^4$ is not locally integrable and requires regularization.
We understand it as analytic continuation of the functional $r^\lm$ to $\lm=-4$,
which is mentioned at the beginning.

The equality
\begin{equation*}
  \triangle r^\mu=\frac1{r^2}\pl_r\left(r^2\pl_r r^\mu\right)
  =\mu(\mu+1)r^{\mu-2}
\end{equation*}
is simply verified for $\re\mu>2$. It remains valid also after analytic
continuation to the point $\mu=-2$. For $\mu=-2$, it looks as follows
\begin{equation*}
  \triangle\frac1{r^2}=2\frac1{r^4}.
\end{equation*}
Thus, for analytic regularization, we attribute to the locally nonintegrable
function $1/r^4$ the functional
\begin{equation*}
  \left(\frac1{r^4},\vf\right):=\frac12\left(\triangle\frac1{r^2},\vf\right)
  =\frac12\left(\frac1{r^2},\triangle\vf\right).
\end{equation*}
The last functional is well defined.

We see that the nonlinear term on the left hand side of the equation
(\ref{qfilht}) is canceled after integration with a test function
\begin{equation}                                                  \label{qsokhg}
  \left(\triangle\frac{M^2}{4r^2}-\frac{M^2}{2r^4},\vf\right)=0,
\end{equation}
if the analytic regularization is used for the function $M^2/2r^4$.
\end{proof}

Equation (\ref{qdytwv}) for a massive point particle was considered in [2--4]
for Weyl Euclidean form of the metric. The authors used parameterization
$f:=\chi^2$ and obtained nonlinear equation
\begin{equation}                                                  \label{qteral}
  \chi\triangle\chi=-2\pi M\dl(\Bx)
\end{equation}
which is equivalent to equation (\ref{qfilht}). They proposed the following
solution
\begin{equation}                                                  \label{qrqwsd}
  \chi(r)=1+\frac M{2r\chi(0)}
\end{equation}
differing from that corresponding to solution (\ref{qsolco}) by the essential
factor $\chi(0)$. This factor is obtained from the regularized (smeared)
$\dl$-function and depends on small parameter $\e$,
\begin{equation*}
  \chi(0)\sim \sqrt{\frac M \e},
\end{equation*}
which is the ``radius'' of the regularized $\dl$-function. In the limit
$\e\to0$, solution (\ref{qrqwsd}) becomes trivial, $\chi=1$, and satisfies
equation $\sqrt g R=0$ corresponding to a vacuum rather than to a point
particle. So the solution proposed in [2--4] does not depend on mass $M$. The
generalization of equation (\ref{qteral}) for a charged massive particle
interacting with electromagnetic field was also considered. In this case, the
factor $\chi(0)$ becomes nontrivial, and the total mass of a particle is
proportional to its charge. This effect is interpreted as the regularization of
the self-energy of a point charge by gravitational interaction.

To show the difference in the approaches, let us consider equation
(\ref{qteral}) for
\begin{equation}                                                  \label{qkutef}
  \chi=1+\frac M{2r},
\end{equation}
which is different from (\ref{qrqwsd}). Substitution of this solution into
equation (\ref{qteral}) yields
\begin{equation}
  \left(1+\frac M{2r}\right)\triangle\left(1+\frac M{2r}\right)=
  \left(1+\frac M{2r}\right)(-2\pi M)\dl(\Bx)=-2\pi M\dl(\Bx).
\end{equation}
Expression $\frac1r\dl(\Bx)$ is not defined. We define it as follows. Relation
\begin{equation*}
  \triangle\frac{r^\lm}r=\lm(\lm-1)r^{\lm-3}+r^\lm\triangle\frac1r
\end{equation*}
can be easily verified for $\re\lm>3$. We analytically continue it to the point
$\lm=-1$ where it can be rewritten as
\begin{equation*}
  \frac1r\triangle\frac1r=\triangle\frac1{r^2}-\frac2{r^4}.
\end{equation*}
So, we define
\begin{equation*}
  \left(\frac1r\dl(\Bx),\vf\right)
  :=-\frac1{4\pi}\left(\frac1r\triangle\frac1r,\vf\right)
  :=-\frac1{4\pi}\left(\triangle\frac1{r^2}-\frac2{r^4},\vf\right).
\end{equation*}
For analytically regularized $\frac1{r^4}$ the right hand side is zero, and we
get relation
\begin{equation*}
  \frac1r\dl(\Bx)=0
\end{equation*}
which is valid in a generalized sense. In other words, the relation
$r^\lm\dl(\Bx)=0$ which is valid for $\lm>0$ is analytically continued to the
point $\lm=-1$.

We see that solution (\ref{qsolco}) can be obtained either by solving equation
(\ref{qfilht}) or (\ref{qteral}). In both cases, we used the analytic
regularization and get the same solution. Our solution (\ref{qsolco}) does
depend on mass $M$. It is nontrivial for neutral particle and differs
essentially from solution (\ref{qrqwsd}) considered in [2--4].

Thus we have found the generalized solution of equation (\ref{qfilht}) in the
space of functionals $\CD'(\MR^3)$. For this solution the left hand side of the
main equation is well defined. Unfortunately, we did not describe explicitly
that subspace in $\CD'(\MR^3)$, for which the left hand side of the main
equation is defined in general. This question is complicated and related to the
uniqueness of the obtained solution. We leave it for future work.
\section{Last equations}
We are left with equations (\ref{qsemby}) and (\ref{qlpojg}) which are to be
solved.

As for the case of equation (\ref{qdytwv}), we cannot explicitly describe the
subspace in $\CD'(\MR^3)$ for which the nonlinear part of equation
(\ref{qsemby}) is defined. Instead, we write down the solution in explicit form
and check that all terms are really defined for it.
\begin{theorem}
Let three dimensional metric has the form (\ref{qfsdok}) where the conformal
factor is given by equation (\ref{qsolco}). Then lapse function
\begin{equation}                                                  \label{qshuyr}
  N=\frac{1-\frac M{2r}}{1+\frac M{2r}}.
\end{equation}
satisfies equation (\ref{qsemby}) in $\MR^3$ in a generalized sense.
\end{theorem}
\begin{proof}
The lapse function is positive for $r>M/2$ and negative for $0<r<M/2$. It has
finite limit at zero,
\begin{equation*}
  \underset{r\to0}\lim N=-1,
\end{equation*}
and is infinitely differentiable for $r>0$. Therefore the lapse function is
locally integrable and lies in $\CD'(\MR^3)$.

One can easily verify that the Laplacian of the lapse function for $r>0$ is
equal to zero,
\begin{equation*}
  \hat\triangle N=\frac1{\hat e}\pl_\mu\big(\hat e\hat g^{\mu\nu}\pl_\nu N\big)
  =0.
\end{equation*}
Function $N$ is continuous, and its derivatives have discontinuities at zero.
Since $N\in\CD'(\MR^3)$ then the weak derivatives are defined everywhere, and
equality
\begin{equation*}
  \hat\triangle N=0,
\end{equation*}
is fulfilled in a generalized sense in the whole $\MR^3$ because the lapse
function has finite limit at zero. This equality can be multiplied by
$\hat eg^{\mu\nu}$ because near zero
\begin{equation}                                                  \label{qhtser}
  \hat e\sim r^{-4},\qquad g^{\mu\nu}\sim r^4\dl^{\mu\nu}.
\end{equation}
Thus the first term on the left hand side of equation (\ref{qsemby}) is equal to
zero in the generalized sense in the whole $\MR^3$.

One can verify that for $r>0$ the following equality holds
\begin{equation*}
  \hat e\hat\nb^\mu\hat\nb^\nu N
  +\hat eN\left(\hat R^{\mu\nu}-\frac12g^{\mu\nu}\hat R\right)=0.
\end{equation*}
Hence we are left only with $\dl$-functions which are contained in the Ricci
tensor and scalar curvature. They do not cancel, and equation (\ref{qsemby})
takes the form
\begin{equation*}
  \frac14 g^{\mu\nu}N\hat e\hat R=4\pi M g^{\mu\nu} N\dl(\Bx)=0.
\end{equation*}
This equation is fulfilled because near zero
\begin{equation*}
  g^{\mu\nu}N\sim \dl^{\mu\nu}r^4,\quad\text{and}\quad r^4\dl(\Bx)=0.
\end{equation*}
Thus equation (\ref{qsemby}) is fulfilled.
\end{proof}

We are left only with equation (\ref{qlpojg}) for the point particle. Since
\begin{equation*}
  \pl_r N=\frac M{r^2}\frac1{\left(1+\frac M{2r}\right)^2},
\end{equation*}
it is not satisfied. At this point, we invoke physical arguments. The right hand
side of equation (\ref{eqmopg}) is the force acting on the particle from the
gravitational field which is produced by the particle itself. Therefore equation
(\ref{eqmopg}) is understood after averaging over a sphere surrounding a
particle. Then it is fulfilled because of the spherical symmetry. In other
words, a particle does not move under the action of its own gravitational field.

The same situation happens in Newton's gravity and electrodynamics. For example,
for gravitational potential $\vf=-1/r$ in Newton's theory, the force acting on a
particle from its own gravitational field is equal to
\begin{equation*}
  \pl_\mu\frac1r=\frac{x_\mu}{r^3},
\end{equation*}
and equations of motion are not satisfied. Nevertheless they are fulfilled after
averaging over a sphere.
\section{Relation to the Schwarzschild solution}
Taking together the laps function (\ref{qshuyr}) and the conformal factor for
the three-dimensional space part of metric (\ref{qsolco}) we obtain the
metric
\begin{equation}                                                  \label{qschre}
  ds^2=\left(\frac{1-\frac M{2r}}{1+\frac M{2r}}\right)^2dt^2
  -\left(1+\frac M{2r}\right)^4
  \left[dr^2+r^2(d\theta^2+\sin^2\theta d\vf^2)\right].
\end{equation}
It is the well known Schwarzschild metric in isotropic coordinates. Isotropic
coordinates for the Schwarzschild metric are well known for a long
time, see, i.e.\ \cite{LanLif62}. The new result is not the Schwarzschild metric
in isotropic coordinates (\ref{qschre}), but the statement that it satisfies
equations (\ref{qdytwv}) and (\ref{qsemby}) with $\dl$-type source.
This fact has principal meaning and far reaching physical consequences. We shall
see that gravitational attraction to the point mass $M$ at large distances
is replaced by repulsion at distances $r<M/2$.

For $M>0$, metric (\ref{qschre}) is defined everywhere in $\MR^4$ except the
world line of the origin of the spherical coordinate system:
\begin{equation*}
  -\infty<t<\infty,~~0<r<\infty,~~0\le\theta\le\pi,~~0\le\vf<2\pi.
\end{equation*}

Let us compare the Schwarzschild metric in isotropic coordinates (\ref{qschre})
and the \linebreak [4] Schwarzschild metric in Schwarzschild coordinates
\cite{Schwar16}
\begin{equation}                                                  \label{qschrt}
  ds^2=\left(1-\frac{2M}\rho\right)dt^2
  -\frac{d\rho^2}{1-\displaystyle\frac{2M}\rho}
  -\rho^2(d\theta^2+\sin^2\theta d\vf^2).
\end{equation}
The Schwarzschild radial coordinate is denoted here by $\rho$. The last metric
is defined for $2M<\rho<\infty$ (outside the horizon $\rho_s=2M$) and for
$0<\rho<2M$ (under the horizon). It is asymptotically flat and tends to the
Lorentz metric when $\rho\to\infty$. It was obtained without $\dl$-function
source, and the constant $M$ appears in it as the integration constant. From
mathematical point of view, it may take arbitrary values $M\in\MR$. However, if
we assume that the Schwarzschild solution describes the gravitational field
outside a point particle, then comparison with the Newton gravitational law at
large distances (see, i.e.\ \cite{LanLif62}) tells us that the integration
constant $M$ is the mass of a particle and therefore must be positive.

The transformation of radial coordinate which brings metric (\ref{qschrt})
into the form (\ref{qschre}) is well known
\begin{equation}                                                  \label{qincde}
  \rho=r\left(1+\frac M{2r}\right)^2.
\end{equation}
It is shown in Fig.\ref{fcotra}. When radius $r$ decreases from infinity to the
critical value $r_*:=M/2$, the Schwarzschild radial coordinate decreases from
$\infty$ to the horizon $\rho_s:=2M$. Afterwards the radius $\rho$ increases
from $\rho_s$ to $\infty$ as the radius $r$ decreases from $r_*$ to zero. Thus
two copies of the Schwarzschild metric outside horizon $2M<\rho<\infty$ are
mapped on two distinct domains in $\MR^3$: $0<r<r_*$ and $r_*<r<\infty$. On the
critical sphere $r=r_*$ they are smoothly glued together. The spatial part of
the metric is not degenerate here. Note that area of a sphere surrounding a
particle tends to infinity as it gets closer to it. This is related to the fact
that components of the spatial part of the metric (\ref{qschre}) diverge for
$r\to0$.
\begin{figure}[h,b,t]
\hfill\includegraphics[width=.35\textwidth]{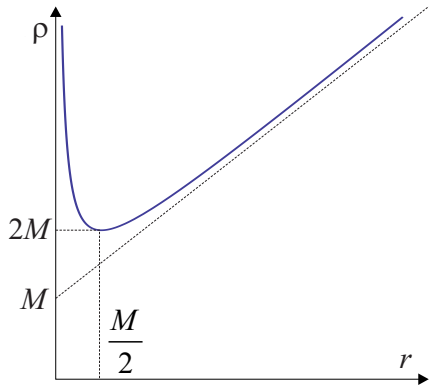}
\hfill {}
\centering \caption{Transformation of the Schwarzschild coordinates to the
isotropic ones.}
\label{fcotra}
\end{figure}

One can easily calculate the asymptotic of the zeroth component of the metric
(\ref{qschre})
\begin{equation*}
  g_{00}\approx 1-\frac{2M}r,\qquad r\to\infty.
\end{equation*}
It is the same as for the Schwarzschild metric in the Schwarzschild coordinates.
This is not surprising because $\rho\to r$ for $r\to\infty$. It means that there
is no problem with Newton's gravitational law for metric (\ref{qschre}) because
it is defined by the asymptotic at $r\to\infty$.

We see that point massive particle is located at that site where the
Schwarzschild radial coordinate equals infinity. This phenomena may be called
gravitational asymptotic freedom because near the particle,
$r=0\Rightarrow\rho=\infty$ the space-time becomes asymptotically flat.

The determinant of metric (\ref{qschre}) is equal to
\begin{equation}                                                  \label{qdetfr}
  \det g_{\al\bt}=-\left(1-\frac M{2r}\right)^2\left(1+\frac M{2r}\right)^{10}
  r^4\sin^2\theta.
\end{equation}
Consequently, the metric is degenerate on the sphere of radius
\begin{equation}                                                  \label{qdefra}
  r_*:=\frac M2,
\end{equation}
which corresponds to the horizon $\rho_s=2M$ for the Schwarzschild solution and
on the $z$-axis $(\theta=0,\pi)$. Degeneracy on the $z$-axis is related to
the spherical coordinate system and location of a point particle at the origin
of the coordinate system for $z=0$.

It is well known that all geometric invariants for the Schwarzschild metric
constructed from the curvature tensor have no singularities on the horizon.
Thus, we have no curvature singularities at the critical sphere $r=M/2$.
Nevertheless, there is a problem.

To describe global structure of the obtained space-time we have to analyze
behaviour of extremals (geodesics) for Schwarzschild solution in isotropic
coordinates which describe motion of probe particles. The behavior of geodesics
in Schwarzschild coordinates is well known and can be recalculated in
isotropic coordinates. Therefore we mention here only some points of their
qualitative behavior. The space-time at $r=0$ (near the particle) and at
$r=\infty$ (space infinity) is geodesically complete, because both regions
correspond to infinite value of the Schwarzschild radial coordinate
$\rho=\infty$. At critical value of the radius $r=r_*$ corresponding to the
Schwarzschild radius $\rho_s$ geodesics are well known to be incomplete. It
means that at $r=r_*$ the space-time in isotropic coordinates is geodesically
incomplete, and we have to extend it. The best way for isotropic coordinates is
to identify geodesics from outside and inside the critical sphere $r_*$.
Afterwards we would get geodesically complete space-time. The problem is how to
do this in a mathematically consistent way. This is a hard problem, and we hope
to address it in the future.

Another particular property of isotropic coordinates is that gravitational
attraction at large distances changes to repulsion. It is clear because in
Schwarzschild coordinates, probe particles always have acceleration directed
towards horizon. It means that probe particles located inside the critical
sphere move outside the center. Thus the gravitational attraction outside
$r_*$ changes to repulsion in the region $0<r<r_*$. As far as we know this is a
novel feature of the obtained solution.
\section{Conclusion}
In the present paper, we prove that the Schwarzschild solution in isotropic
coordinates satisfies the system of Einstein's and geodesic equations for a
point massive particle. Nontrivial energy-momentum tensor appears on the right
hand side of Einstein's equations which is proportional to three-dimensional
$\dl$-function. Rewriting the Schwarzschild solution in isotropic coordinates is
by itself a student's problem. The result of the paper is that we attribute
mathematical meaning for Einstein's equations with $\dl$-type source and provide
explicit solution.

The mass $M$ appears in the Schwarzschild solution as an integration constant
and may be arbitrary. In our approach, the mass $M$ of a point particle enters
the action from the very beginning.

The obtained solution is isometric to Einstein--Rosen bridge \cite{EinRos35}.
In the original paper, the bridge was attributed to an elementary particle which
in our notations corresponds to the critical sphere $r=r_*$. The space-time is
described by two sheets. We have shown that the particle described by a
$\dl$-function is located not at $r=r_*$, but at geodesically complete
``infinity'' of one of the sheets corresponding to $r=0$. Geodesically complete
infinity of the second sheet lies at infinity $r=\infty$ and corresponds to
asymptotically flat space-time. Both sheets are glued together at the critical
sphere $r=r_*$ where metric becomes degenerate and gravitational attraction at
large distances is replaced by repulsion at small distances. This effect is of
primary importance and makes us to reconsider our approach to many gravitational
phenomena.

At present, the Einstein--Rosen bridge has another interpretation. Two sheets
are considered as two different universes which are connected by a worm hole at
$r=r_*$. It was noted recently \cite{GuKaNiPa09} that nontrivial energy-momentum
tensor appears on the right hand side of Einstein's equations at $r_*$. In
our approach, the right hand side of Einstein's equations is equal to zero at
$r=r_*$. The difference comes from choosing the solution to equation
$N^2=g_{00}$. For a given metric we chose solution (\ref{qshuyr}) which is
positive for $r>r_*$ and negative at $r<r_*$. The other solution $|N|$ is
chosen in \cite{GuKaNiPa09}, the modulus sign resulting in the appearance of the
singularity on the right hand side of Einstein's equations at $r=r_*$.

The obtained space-time is topologically four-dimensional Euclidean space with
removed world line of the particle $r=0$. This space-time is geodesically
complete at $r\to\infty$ as well as at $r\to 0$. At the critical sphere $r_*$
geodesics are incomplete. The problem of how to continue them remains open.

The effect of the transformation of gravitational attraction into repulsion is
the straightforward consequence of Einstein's equations and leads to many
questions. For example, how can a black hole be formed if particles cannot
come close one to another ? One can make a lot of speculations at this point but
we do not consider them here.

The author is grateful to I.~V.~Volovich, A.~K.~Gushchin, Yu.~N.~Drozhzhinov,
\linebreak[4]\fbox{B.~I.~Zavialov}, and V.~P.~Mikhailov for discussions and
valuable comments. The work is supported by RFFI (grants 11-01-00828-a and
13-01-12424-ofi-m), the Program for Supporting Leading Scientific Schools
(grant NSh-2928.2012.1), and the Program ``Modern problems of theoretical
mathematics'' by RAS.

\end{document}